\documentclass[nopacs,pra,11pt,longbibliography,footinbib]{revtex4}
 \usepackage{times}
 \usepackage{array}   
\usepackage{algorithm}
\usepackage[noend]{algpseudocode}
\usepackage{caption}
\usepackage[paperwidth=199.8mm,paperheight=297mm,centering,hmargin=15mm,vmargin=2cm]{geometry}
\usepackage[dvipsnames]{xcolor}
\usepackage{framed}
\definecolor{shadecolor}{rgb}{0.9,0.9,0.9}
\usepackage{mathtools}
\usepackage{amsmath}
\usepackage[shortlabels]{enumitem}
\usepackage[titletoc,title]{appendix}
\usepackage{graphicx,epic,eepic,epsfig,amsmath,latexsym,amssymb,verbatim,color}
\usepackage{dsfont}

\usepackage{float}
\usepackage{tikz}
\usetikzlibrary{chains}
\usetikzlibrary{fit}
\usepackage{pgflibraryarrows}		%optional
\usepackage{pgflibrarysnakes}		%optional

\usepackage{epsfig}
\usetikzlibrary{shapes.symbols,patterns} % for source symbols
\usepackage{pgfplots}

\usepackage[strict]{changepage}
\usepackage{hyperref}
\hypersetup{colorlinks=true,citecolor=blue,linkcolor=blue,filecolor=blue,urlcolor=blue,breaklinks=true}

\usepackage[marginal]{footmisc}
\usepackage{url}
\usepackage{theorem}
\newtheorem{definition}{Definition}
\newtheorem{proposition}[definition]{Proposition}

\newtheorem{corollary}[definition]{Corollary}

\usepackage{colortbl}
\usepackage{pifont}% http://ctan.org/pkg/pifont
\definecolor{Gray}{gray}{0.92}
\definecolor{Gray2}{gray}{0.75}
\definecolor{maroon}{cmyk}{0,0.87,0.68,0.32}
\usepackage{booktabs} 

\def\squareforqed{\hbox{\rlap{$\sqcap$}$\sqcup$}}
\def\qed{\ifmmode\squareforqed\else{\unskip\nobreak\hfil
\penalty50\hskip1em\null\nobreak\hfil\squareforqed
\parfillskip=0pt\finalhyphendemerits=0\endgraf}\fi}
\def\endenv{\ifmmode\;\else{\unskip\nobreak\hfil
\penalty50\hskip1em\null\nobreak\hfil\;
\parfillskip=0pt\finalhyphendemerits=0\endgraf}\fi}
\newenvironment{proof}{\noindent \textbf{{Proof~} }}{\hfill $\blacksquare$}

\newcounter{remark}
\newenvironment{remark}[1][]{\refstepcounter{remark}\par\medskip\noindent%
\textbf{Remark~\theremark #1} }{\medskip}

\newcounter{example}

\mathchardef\ordinarycolon\mathcode`\:
\mathcode`\:=\string"8000
\def\vcentcolon{\mathrel{\mathop\ordinarycolon}}
\begingroup \catcode`\:=\active
  \lowercase{\endgroup
  \let :\vcentcolon
  }

\usepackage{cleveref}
\usepackage{graphicx}
\usepackage{xcolor}

\RequirePackage[framemethod=default]{mdframed}
\newmdenv[skipabove=7pt,
skipbelow=7pt,
backgroundcolor=darkblue!15,
innerleftmargin=5pt,
innerrightmargin=5pt,
innertopmargin=5pt,
leftmargin=0cm,
rightmargin=0cm,
innerbottommargin=5pt,
linewidth=1pt]{tBox}

\newmdenv[skipabove=7pt,
skipbelow=7pt,
backgroundcolor=blue2!25,
innerleftmargin=5pt,
innerrightmargin=5pt,
innertopmargin=5pt,
leftmargin=0cm,
rightmargin=0cm,
innerbottommargin=5pt,
linewidth=1pt]{dBox}
\newmdenv[skipabove=7pt,
skipbelow=7pt,
backgroundcolor=darkkblue!15,
innerleftmargin=5pt,
innerrightmargin=5pt,
innertopmargin=5pt,
leftmargin=0cm,
rightmargin=0cm,
innerbottommargin=5pt,
linewidth=1pt]{sBox}
\definecolor{darkblue}{RGB}{0,76,156}
\definecolor{darkkblue}{RGB}{0,0,153}
\definecolor{blue2}{RGB}{102,178,255}
\definecolor{darkred}{RGB}{195,0,0}
\newcommand{\nc}{\newcommand}
\nc{\rnc}{\renewcommand}
\nc{\beg}{\begin{equation}}
\nc{\eeq}{{\end{equation}}}
\nc{\beqa}{\begin{eqnarray}}
\nc{\eeqa}{\end{eqnarray}}
\nc{\lbar}[1]{\overline{#1}}
\nc{\bra}[1]{\langle#1|}
\nc{\ket}[1]{|#1\rangle}
\nc{\ketbra}[2]{|#1\rangle\!\langle#2|}
\nc{\braket}[2]{\langle#1|#2\rangle}

\nc{\proj}[1]{| #1\rangle\!\langle #1 |}
\nc{\avg}[1]{\langle#1\rangle}
\nc{\Rank}{\operatorname{Rank}}
\nc{\smfrac}[2]{\mbox{$\frac{#1}{#2}$}}
\nc{\tr}{\operatorname{Tr}}
\nc{\ox}{\otimes}
\nc{\dg}{\dagger}
\nc{\dn}{\downarrow}
\nc{\cA}{{\cal A}}
\nc{\cB}{{\cal B}}
\nc{\cC}{{\cal C}}
\nc{\cD}{{\cal D}}
\nc{\cE}{{\cal E}}
\nc{\cF}{{\cal F}}
\nc{\cG}{{\cal G}}
\nc{\cH}{{\cal H}}
\nc{\cI}{{\cal I}}
\nc{\cJ}{{\cal J}}
\nc{\cK}{{\cal K}}
\nc{\cL}{{\cal L}}
\nc{\cM}{{\cal M}}
\nc{\cN}{{\cal N}}
\nc{\cO}{{\cal O}}
\nc{\cP}{{\cal P}}
\nc{\cQ}{{\cal Q}}
\nc{\cR}{{\cal R}}
\nc{\cS}{{\cal S}}
\nc{\cT}{{\cal T}}
\nc{\cV}{{\cal V}}
\nc{\cX}{{\cal X}}
\nc{\cY}{{\cal Y}}
\nc{\cZ}{{\cal Z}}
\nc{\cW}{{\cal W}}
\nc{\csupp}{{\operatorname{csupp}}}
\nc{\qsupp}{{\operatorname{qsupp}}}
\nc{\var}{{\operatorname{var}}}
\nc{\rar}{\rightarrow}
\nc{\lrar}{\longrightarrow}
\nc{\polylog}{{\operatorname{polylog}}}
\nc{\1}{{\mathds{1}}}
\nc{\wt}{{\operatorname{wt}}}
\nc{\av}[1]{{\left\langle {#1} \right\rangle}}
\nc{\supp}{{\operatorname{supp}}}

\def\a{\alpha}

\def\di{\diamondsuit}

\def\ve{\varepsilon}

\def\x{\xi}

\nc{\RR}{{{\mathbb R}}}
\nc{\CC}{{{\mathbb C}}}
\nc{\FF}{{{\mathbb F}}}
\nc{\NN}{{{\mathbb N}}}
\nc{\ZZ}{{{\mathbb Z}}}
\nc{\PP}{{{\mathbb P}}}
\nc{\QQ}{{{\mathbb Q}}}
\nc{\UU}{{{\mathbb U}}}
\nc{\EE}{{{\mathbb E}}}
\nc{\id}{{\operatorname{id}}}

\nc{\CHSH}{{\operatorname{CHSH}}}

\nc{\be}{\begin{equation}}
\nc{\ee}{{\end{equation}}}
\nc{\bea}{\begin{eqnarray}}
\nc{\eea}{\end{eqnarray}}
\nc{\<}{\langle}
\rnc{\>}{\rangle}
\nc{\rU}{\mbox{U}}

\nc{\ob}[1]{#1}

\nc{\SEP}{{\text{\rm SEP}}}
\nc{\NS}{{\text{\rm NS}}}
\nc{\LOCC}{{\text{\rm LOCC}}}
\nc{\PPT}{{\text{\rm PPT}}}
\nc{\EXT}{{\text{\rm EXT}}}
\nc{\Sym}{{\operatorname{Sym}}}

\nc{\ERLO}{{E_{\text{r,LO}}}}
\nc{\ERLOCC}{{E_{\text{r,LOCC}}}}
\nc{\ERPPT}{{E_{\text{r,PPT}}}}
\nc{\ERLOCCinfty}{{E^{\infty}_{\text{r,LOCC}}}}
\nc{\Aram}{{\operatorname{\sf A}}}

\usepackage{tikz}
\usepackage{hyperref}
\hypersetup{colorlinks=true,citecolor=blue,linkcolor=blue,filecolor=blue,urlcolor=blue,breaklinks=true}

\makeatletter
\def\grd@save@target#1{%
  \def\grd@target{#1}}
\def\grd@save@start#1{%
  \def\grd@start{#1}}
\tikzset{
  grid with coordinates/.style={
    to path={%
      \pgfextra{%
        \edef\grd@@target{(\tikztotarget)}%
        \tikz@scan@one@point\grd@save@target\grd@@target\relax
        \edef\grd@@start{(\tikztostart)}%
        \tikz@scan@one@point\grd@save@start\grd@@start\relax
        \draw[minor help lines,magenta] (\tikztostart) grid (\tikztotarget);
        \draw[major help lines] (\tikztostart) grid (\tikztotarget);
        \grd@start
        \pgfmathsetmacro{\grd@xa}{\the\pgf@x/1cm}
        \pgfmathsetmacro{\grd@ya}{\the\pgf@y/1cm}
        \grd@target
        \pgfmathsetmacro{\grd@xb}{\the\pgf@x/1cm}
        \pgfmathsetmacro{\grd@yb}{\the\pgf@y/1cm}
        \pgfmathsetmacro{\grd@xc}{\grd@xa + \pgfkeysvalueof{/tikz/grid with coordinates/major step}}
        \pgfmathsetmacro{\grd@yc}{\grd@ya + \pgfkeysvalueof{/tikz/grid with coordinates/major step}}
        \foreach \x in {\grd@xa,\grd@xc,...,\grd@xb}
        \node[anchor=north] at (\x,\grd@ya) {\pgfmathprintnumber{\x}};
        \foreach \y in {\grd@ya,\grd@yc,...,\grd@yb}
        \node[anchor=east] at (\grd@xa,\y) {\pgfmathprintnumber{\y}};
      }
    }
  },
  minor help lines/.style={
    help lines,
    step=\pgfkeysvalueof{/tikz/grid with coordinates/minor step}
  },
  major help lines/.style={
    help lines,
    line width=\pgfkeysvalueof{/tikz/grid with coordinates/major line width},
    step=\pgfkeysvalueof{/tikz/grid with coordinates/major step}
  },
  grid with coordinates/.cd,
  minor step/.initial=.2,
  major step/.initial=1,
  major line width/.initial=2pt,
}
\makeatother

\usepackage{thmtools}
\usepackage{thm-restate}
\usepackage{etoolbox}
\makeatletter
\def\problem@s{}
\newcounter{problems@cnt}

\newcommand{\allproblems}{\problem@s}
\makeatother

\definecolor{colorone}{rgb}{1,0.36,0.03}
\definecolor{colortwo}{rgb}{0.54,0.71,0.03}
\definecolor{colorthree}{rgb}{0.01,0.51,0.93}
\definecolor{colorfour}{rgb}{0.47,0.26,0.58}
\usepackage{tcolorbox}
\usepackage{relsize}
\usepackage{graphicx}
\nc{\st}{\text{subject to} \ }
\nc{\supre}{\text{supremum} \ }
\nc{\sdp}{\text{sdp}}
\nc{\cU}{\mathcal U}

\begin{document}
%\title{\textbf{Optimizing the fundamental limits for quantum and private communication}}

\title{\textbf{Pursuing the fundamental limits for quantum  communication}}

\author{Xin Wang}
\affiliation{Institute for Quantum Computing, Baidu Research, Beijing 100193, China}
\email{wangxin73@baidu.com}
   
\date{\today}

\begin{abstract}
The quantum capacity of a noisy quantum channel determines the maximal rate at which we can code reliably over asymptotically many uses of the channel, and it characterizes the channel's ultimate ability to transmit quantum information coherently. In this paper, we derive single-letter upper bounds on the quantum and private capacities of quantum channels. The quantum capacity of a quantum channel is always no larger than the quantum capacity of its extended channels, since the extensions of the channel can be considered as assistance from the environment. By optimizing the  parametrized extended channels with specific structures such as the flag structure, we obtain new upper bounds on the quantum capacity of the original quantum channel. Furthermore, we extend our approach to estimating the fundamental limits of private communication and one-way entanglement distillation. As notable applications, we establish improved upper bounds to the quantum and private capacities for fundamental quantum channels of interest in quantum information, some of which are also the sources of noise in superconducting quantum computing. In particular, our upper bounds on the quantum capacities of the depolarizing channel and the generalized amplitude damping channel are strictly better than previously best-known bounds for certain regimes. 
\end{abstract}

\maketitle

% \tableofcontents
\section{Introduction}

Quantum communication is an integral part of quantum information theory. The quantum capacity of a noisy quantum channel is the maximum rate at which it can convey quantum information reliably over asymptotically many uses of the channel. The theorem by Lloyd, Shor, and Devetak~\cite{Lloyd1997,Shor2002a,Devetak2005a} and the work in Refs.~\cite{Schumacher1996a,Barnum2000,Barnum1998} show
that the quantum capacity is equal to the regularized coherent information, i.e., the quantum capacity of a noisy quantum channel $\cN$
is given by 
\begin{equation}
  Q(\cN) = \lim_{n\to\infty} 
  \frac{Q^{(1)}(\cN^{\otimes n})}{n},
\label{eq:qcap}
\end{equation}
where $Q^{(1)}(\cN):= \max_{\rho_{A}} {H}(\cN(\rho_{A})) - H(\cN^c(\rho_{A}))$ denotes the coherent information of the channel,
 $\cN^c$ is the complementary channel of $\cN$, and $H$ is
the von~Neumann entropy.

Quantum capacity is notoriously complicated and hard to evaluate since it is characterized by the above multi-letter, regularized expression. Note that such  regularization is necessary in general due to the super-additivity of channel's coherent information \cite{DiVincenzo1998a,Leditzky2018a} and it is worse that an unbounded number of channel uses may be required to detect a channel’s capacity \cite{Cubitt2015}. 

To evaluate the quantum capacity, substantial efforts have been made in the past two decades. During the development of estimating the quantum capacity, there are two threads of studying the upper bounds on quantum capacity. One is to develop single-letter upper bounds that are efficiently computable for general quantum channels~(see,~e.g.,~\cite{Holevo2001,Muller-Hermes2015,Wang2016a,Wang2017d,Berta2017a,XWthesis,Pisarczyk2018,Fang2019a}), which aims to provide efficiently computable benchmarks for arbitrary quantum noise. Another is to develop single-letter upper bounds that are relatively tight or computable for specific classes of quantum channels (see,~e.g.,~\cite{Wolf2007,Smith2008a,Sutter2014,Cerf2000, Smith2008b,Leditzky2017,Fanizza2019,Gao2015a,Tomamichel2015a}), which help us better understand quantum communication via these quantum channels. 

The depolarizing channel $\cD_p^d(\rho):=(1-p)\rho + p \tr(\rho)I_d/d$ is one of the most important and fundamental quantum channels \cite{Nielsen2010,Wilde2017book}, which is useful in modelling noise for quantum hardware such as the superconducting quantum processor~\cite{Arute2019}. However, even for the qubit depolarizing channel, the quantum capacity remains unsolved despite substantial efforts \cite{Smith2008a,Cerf2000,Sutter2014,Smith2008b,Leditzky2017,Leditzky2018a, Fanizza2019}. Recently, Fanizza et al.~\cite{Fanizza2019} considered the degradable extensions with non-orthogonal quantum flags and obtained an improved upper bound on the quantum capacity of qubit depolarizing channel in an intermediate regime of noise. 

In this paper, we establish single-letter converse bounds on the quantum and private capacities of a noisy quantum channel. As notable applications, we establish improved upper bounds on the quantum capacity for fundamental channels of great interest in quantum communication, including the depolarizing channel and the generalized amplitude damping channel. In particular, we also obtain efficiently computable upper bounds on the one-way distillable entanglement of general bipartite quantum states.

\section{Results}
\subsection{Upper bounds on one-way distillable entanglement}
We begin with the preliminaries on quantum information. We will frequently use symbols such as $A$ (or $A'$) and $B$ (or $B'$) to denote Hilbert spaces associated with Alice and Bob, respectively. We use $d_A$ to denote the dimension of system $A$. The set of linear operators acting on $A$ is denoted by $\cL(A)$. We usually write an operator with a subscript indicating the system that the operator acts on, such as $M_{AB}$, and write $M_A:=\tr_B M_{AB}$. Note that for a linear operator $X\in\cL(A)$, we define $|X|=\sqrt{X^\dagger X}$, where $X^\dagger$ is the adjoint operator of $X$, and the trace norm of $X$ is given by $\|X\|_1=\tr |X|$.

Quantum information processing tasks necessarily rely on entanglement resources, especially the maximally entangled states. It is thus important to develop entanglement distillation protocols, which aims to obtain maximally entangled states from less-entangled bipartite states shared between two parties using one-way or two-way local operations and classical communication (LOCC). Let Alice and Bob initially share $n$ copies of a mixed bipartite state $\rho_{AB}$. Their goal is to obtain a state that is close to $\Phi^{\ox {m_n}}$ within some error tolerance $\ve$ via one-way or two-way LOCC, where $\ket{\Phi} = \frac{1}{\sqrt 2}(\ket{00} + \ket{11})$. When $\ve$ vanishes asymptotically, the asymptotic rate at which ebits can be generated is defined by $\lim_{n\to\infty}m_n/n$, and it is called an achievable rate for  entanglement distillation.  The one-way distillable entanglement $E_{D,\to}$ is defined as the supremum over all achievable rates under one-way LOCC.  It is known that the one-way distillable entanglement is given by a regularized
formula 
$D_{\to}(\rho)=\lim_{n \to \infty} \frac{1}{n} D^{(1)}\left(\rho^{\otimes n}\right)$  with $
D^{(1)}(\rho_{AB}):=\sup_{\Pi:{AB\to A'B'}} I\left(A^{\prime}\right\rangle B^{\prime})_{\Pi(\rho)}$, where $I(A\rangle B)_{\rho}=H(\rho_B)-H(\rho_{AB})$ and the optimization is over one-way LOCC operations \cite{Devetak2003a}.
Thus it is extremely difficult to evaluate.
Even for the qubit isotropic states $\rho_{f} = (1-f)\Phi + fI/4,$, the one-way distillable entanglement is unsolved.

Our first main result is an efficiently computable upper bound on the one-way distillable entanglement.  Given a bipartite state $\rho_{AB}$ shared between Alice and Bob,
we focus on its extended bipartite quantum state $\rho_{ABF}$, where Alice holds system $A$ and Bob holds system $BF$. A useful fact is that the one-way distillable entanglement of $\rho_{ABF}$ is always no smaller than the original bipartite $\rho_{AB}$ since the extended system $F$ on Bob's side can be discarded. Therefore, any upper bound on $E_{D,\to}(\rho_{ABF})$ yields an upper bound on $E_{D,\to}(\rho_{AB})$. Our method further estimates $E_{D,\to}(\rho_{ABF})$ via the converse bound via the approximate degradability bound on one-way distillable entanglement~\cite{Leditzky2017} and in particular  optimizes the extended states.

Let $\rho_{AB}$ be a bipartite state with purification $\ket\phi_{ABE}$. The state $\rho_{AB}$ is called degradable if there exists an completely positive and trace-preserving map (CPTP) $\cM_{B\to E}$ such that $\cM_{B\to E}(\rho_{AB}) = \tr_B \phi_{ABE}$. For any degradable state $\rho_{AB}$, its one-way distillable entanglement is given by its coherent information \cite{Leditzky2017}, i.e., $E_{D,\to}(\rho_{AB}) = I(A\rangle B)_{\rho}$.
This nice result and its channel version have been extended to the approximate case~\cite{Leditzky2017,Sutter2014}. For any bipartite state $\rho_{AB}$ with purification $\ket\phi_{ABE}$, its degradability parameter \cite{Leditzky2017} is defined as $\eta\left(\rho\right)_{A|B}:=\min _{\mathcal{M}_{B \rightarrow E} \text{ is CPTP}} \frac{1}{2}\left\|\rho_{A E}-\mathcal{M}_{B\to E}\left(\rho_{A B}\right)\right\|_{1}$, which can be computed via semi-definite programming (SDP).
%Let $\rho_{AB}$ be a bipartite state with decomposition $\rho_{AB} = \sum_{j=0}^k \tau_{AB}^j$. For the extended state $\rho_{A BF}=\sum_{j=1}^k \tau_{AB}^j \ox \sigma_F^j$,
We then have the following chain of inequalities:
\begin{align}
E_{D,\to}(\rho_{AB}) & \le E_{D,\to}(\rho_{ABF}) \\
& \le I(A\rangle BF)_{\rho}+4 \eta\left(\rho\right)_{A|BF} \log |\widehat E|+ 2g(\eta\left(\rho\right)_{A|BF}),
\end{align}
where $g\left(p\right):= (1+p)h({p}/{(1+p)})$ is the bosonic entropy, $h\left(p\right):=-p\log p - (1-p)\log(1-p)$
 is the binary entropy, and $\widehat E$ is the system required to purify $\rho_{ABF}$.
The first line follows since Bob can always discard the system $F$.  The second line
follows due to the bound via approximate degradablity (Theorem 2.12 in \cite{Leditzky2017}).   %This slightly simplified representation of the continuity bound is from \cite{Shirokov2017,Khatri2019}.

The next step is to optimize over the extended states, which is not easy since the optimization is not convex. However, we could optimize over the extended states with certain certain structures, such as the flag structure. To be specific,
for the state with substate decomposition $\rho_{AB}=\sum_{j=1}^k \tau_{AB}^j$ with $\tau_{AB}^j\ge 0$ and $\tr \tau_{AB}^j \le 1$ for each $j$, 
its one-way distillable entanglement is upper bounded by the following optimization
\begin{align}\label{eq:state general bound}
E_{D,\to}(\rho_{AB}) \le  \inf_{\rho_{A BF}}
I(A\rangle BF)_{\rho}+4 \eta\left(\rho\right)_{A|BF} \log |\widehat E|+ 2g(\eta\left(\rho\right)_{A|BF}),
\end{align}
where the optimization is over $\rho_{A BF}=\sum_{j=1}^k \tau_{AB}^j \ox \sigma_F^j$ with flag quantum states $\sigma_F^j$  and  $\phi_{ABF\widehat E}$ is the purification of $\rho_{ABF}$. 

As the above bound may still be difficult to compute, we introduce a more tractable bound via adding more restrictions on the flag structure. 
\begin{proposition}
Let $\rho_{AB}$ be a bipartite state with substate decomposition $\rho_{AB} =  \tau_{AB} + \omega_{AB}$, we have	
\begin{align}\label{eq:state pure flag bound}
	 E_{D,\to}(\rho_{AB}) \le  \inf_{0\le \alpha\le 1} s(\tau_{AB} ,\omega_{AB},\alpha),
	 \end{align}
	 with
	 \begin{align} \label{eq:function s}
     s(\tau_{AB} ,\omega_{AB},\alpha):= I(A\rangle BF)_{\rho}+4 \eta\left(\rho\right)_{A|BF} \log |\widehat E|+ 2g(\eta\left(\rho\right)_{A|BF}),
	\end{align}
where $\rho_{ABF}= \tau_{AB} \otimes \proj{\psi_\a} + \omega_{AB} \otimes \proj 0$, $\ket{\psi_\a} = \sqrt \alpha \ket 0+\sqrt{1-\alpha}\ket1$,  and $\phi_{ABF\widehat E}$ is the purification of $\rho_{A BF}$.  
\end{proposition}
Note that the objective function $s(\tau_{AB},\omega_{AB},\alpha)$  can be computed efficiently since the degradability parameter $\eta(\rho)_{A|BF}$ can be computed via SDP. Then we could do brute-force search on $\a\in[0,1]$ to find the optimal solution. Moreover, other choices of the parametrized flag structure also yield upper bounds on the one-way distillable entanglement.

 \subsection{Improved upper bounds for the depolarizing channel}
We further extend the above results on one-way distillable entanglement to the evaluation of the quantum capacity of the depolarizing channel. Based on the fact that the quantum capacity of a teleportation-simulable channel is equal to the one-way distillable entanglement of its Choi state \cite{Bennett1996c}, we establish new upper bounds on the quantum capacity of the depolarizing channel. Here, a quantum channel  is called teleportation-simulable channel if it can be simulated by the action of a teleportation protocol \cite{Bennett1993,BK98,Werner2001a} on a resource state  shared between the sender and receiver. And the teleportation-simulable channels were also called Choi-stretchable in \cite{Pirandola2015b}, which studied more general simulation protocols consisting of LOCC operations.
Moreover, the Choi state of $\cN_{A\to B}$ is given by $\rho_{\cN} = J_{\cN}/d_A$, where $J_{\cN}=\sum_{ij} \ketbra{i_A}{j_{A}} \ox \cN(\ketbra{i_{A'}}{j_{A'}})$ is the Choi-Jamio\l{}kowski matrix of $\cN$, and $\{\ket{i_A}\}$, $\{\ket{i_{A'}}\}$ are orthonormal bases on isomorphic Hilbert spaces $A$ and $A'$, respectively.
 
For a teleportation-simulable quantum channel $\cN$ with Choi state $\rho_\cN =\sum_{j=1}^k \tau_{AB}^j$ and $\tau_{AB}^j\ge0$ for each $j$, the bound in Eq.~\eqref{eq:state general bound} yields the following upper bound
\begin{align}
Q(\cN)  \le \inf_{\rho_{ABF}} I(A\rangle BF)_{\rho}+4 \eta\left(\rho\right)_{A|BF} \log |E|+ 2g(\eta\left(\rho\right)_{A|BF}),
\end{align}
where $\rho_{A BF}=\sum_{j=1}^k \tau_{AB}^j \ox \sigma_F^j$ and $\sigma_F^j$ are the flag states. Here, $E$ is the system required to purify $\rho_{A BF}$.
%$\rho_{\widehat{\cN}}$ is the Choi state of $\widehat\cN (\cdot)= \sum_{j=0}^k \cN_j(\cdot) \otimes \sigma_j$.  
In particular, for a teleportation-simulable channel with Choi state $\rho_\cN =  \tau_{AB} +  \omega_{AB}$, the efficiently computable bound in Eq.~\eqref{eq:state pure flag bound} leads to 
\begin{align}
	 Q(\cN) \le   \inf_{0\le \alpha\le 1} s(\tau_{AB},\omega_{AB},\alpha), \label{eq:channel state bound}
\end{align}
with 
$s(\tau_{AB} ,\omega_{AB},\alpha)$ defined in Eq.~\eqref{eq:function s}. 

As a notable application, we obtain improved upper bounds on the quantum capacity for the depolarizing channel.
Since the depolarizing channel is teleportation-simulable, we apply our bound in Eq.~\eqref{eq:channel state bound} to evaluate its quantum capacity. 
For the qubit depolarizing channel $\cD_p(\rho):=(1-p)\rho + p \tr(\rho)I_2/2$, its Choi state has the following substate decomposition
\begin{align}
	\rho_{\cD_p} = (1-p)\Phi + pI/4,
\end{align}
with $\ket{\Phi} = \frac{1}{\sqrt 2}(\ket{00} + \ket{11})$ and $I$ is the identity operator. Taking $\tau_{AB} = (1-p)\Phi$ and $\omega_{AB} = pI/4$, we are able to establish the upper bound on the quantum capacity as follows
\begin{align}
Q(\cD_p) \le \inf_{0\le \alpha\le 1} s((1-p)\Phi,pI/4,\alpha). \label{eq:depo bound}
\end{align}

For the low-noise depolarizing channels, our bound outperforms previously best known upper bounds via approximate degradable channels~\cite{Sutter2014} (cf.~Figure.~\ref{fig:depo low noise}). We also note that there is an analytical upper bound on the quantum capacity of the depoarizing channel in the low-noise regime \cite{Leditzky2017a}.
And for the depolarizing channels with intermediate noise, our bound outperforms previously best known upper bounds obtained via considering flag states~\cite{Fanizza2019} as well as the bound obtained via degradable and anti-degradable decompositions~\cite{Leditzky2017} (cf.~Figure.~\ref{fig:depo middle noise}). The code for computing the bound in Eq.~\eqref{eq:depo bound} is available in \cite{Wang_codes}.

\begin{figure}[H]
\centering
  \begin{tikzpicture}
	\begin{axis}[
		height=7.5cm,
		width=9.8cm,
		grid=major,
		xlabel=$3p/4$,
%		ylabel=$\beta$,
		xmin=0,
		xmax=0.02,
		ymax=1.0,
		ymin=0.82,
	    xtick={0,0.005,0.01,0.015,0.02},
        ytick={1,0.98,0.96,0.94,0.92,0.9,0.88,0.86,0.84,0.82},
		legend style={at={(0.67,0.987)},anchor=north,legend cell align=left,font=\footnotesize} 
	]
	
%%%%%%%%%%%%%%%%%%%%%%%%%%%%%%%%%%
% fillings here to not confuse the legend numbering (HAVE TO GO FIRST --> UNDERLAY)
%%%%%%%%%%%%%%%%%%%%%%%%%%%%%%%%%%
%%%%%%%%%%%%%%%%%%%%%%%%%%%%%%%%%%
% actual curves start here!!
%%%%%%%%%%%%%%%%%%%%%%%%%%%%%%%%%%	
%%%%%%%%
		\addplot[NavyBlue,very thick,smooth] coordinates {
  (0,	1)
(0.001,	0.987027121)
(0.002,	0.976248321)
(0.003,	0.966275914)
(0.004,	0.956911126)
(0.005,	0.947941759)
(0.006,	0.939371622)
(0.007,	0.931122385)
(0.008,	0.92315712)
(0.009,	0.91544798)
(0.01,	0.907973172)
(0.011,	0.900715137)
(0.012,	0.893659418)
(0.013,	0.886793875)
(0.014,	0.880108158)
(0.015,	0.873593331)
(0.016,	0.867241589)
(0.017,	0.861046224)
(0.018,	0.855000743)
(0.019,	0.849100015)
(0.02,	0.843338955)

	};
		\addlegendentry{Sutter et al.~\cite{Sutter2014}}	
		
			%old UB
		\addplot[dashed,BrickRed,very thick,smooth] coordinates {
(0,	1)
(0.001,	0.987045866)
(0.002,	0.976123635)
(0.003,	0.965981791)
(0.004,	0.956412704)
(0.005,	0.947254792)
(0.006,	0.938274697)
(0.007,	0.929714078)
(0.008,	0.921399867)
(0.009,	0.913156357)
(0.01,	0.905173372)
(0.011,	0.897435373)
(0.012,	0.889750834)
(0.013,	0.882164705)
(0.014,	0.874872259)
(0.015,	0.867691901)
(0.016,	0.860358221)
(0.017,	0.853431295)
(0.018,	0.846617271)
(0.019,	0.83954677)
(0.02,	0.83293071)

	};
		\addlegendentry{new upper bound~\eqref{eq:depo bound}}

%%%%%%%%%%%%%%%%%%%%%%%%%%%%%%%%%%%%%%%%%%%		
	%lower
		\addplot[dotted,ForestGreen,very thick,smooth] coordinates {
(0,	1)
(0.001,	0.98700728)
(0.002,	0.976016004)
(0.003,	0.965781061)
(0.004,	0.95603779)
(0.005,	0.946660495)
(0.006,	0.937575145)
(0.007,	0.92873282)
(0.008,	0.920098755)
(0.009,	0.911646954)
(0.01,	0.903357239)
(0.011,	0.895213493)
(0.012,	0.88720254)
(0.013,	0.879313412)
(0.014,	0.871536832)
(0.015,	0.863864852)
(0.016,	0.856290589)
(0.017,	0.848808018)
(0.018,	0.841411829)
(0.019,	0.8340973)
(0.02,	0.826860207)
	};
	\addlegendentry{$Q^{(1)}$}
	
%%%%%%%%%%%%%%%%%%%%%%%%%%%%%%%%%%%%%%%%%%%	

%%%%%%%%%%%%%%%%%%%%%%%%%%%%%%%%%%%%%%%%%%%		
	\end{axis}  
%%%%%%%%%%%%%%%%%%%%%%%%%%%%%%%%%%
% labeling here
%%%%%%%%%%%%%%%%%%%%%%%%%%%%%%%%%%

\end{tikzpicture}
\caption{Upper and lower bounds on the quantum capacity $Q(\cD_p)$ of the low-noise qubit depolarizing channel for the interval $3p/4 \in\left[0,0.02\right]$. The channel's coherent information $Q^{(1)}$ (green, dotted) provides a lower bound on the quantum capacity $Q(\cD_p)$. The red dashed line depicts our bound in Eq.~\eqref{eq:depo bound}, which improves the bound via approximate degradable channels obtained in~\cite{Sutter2014}.
		}
\label{fig:depo low noise}
\end{figure}
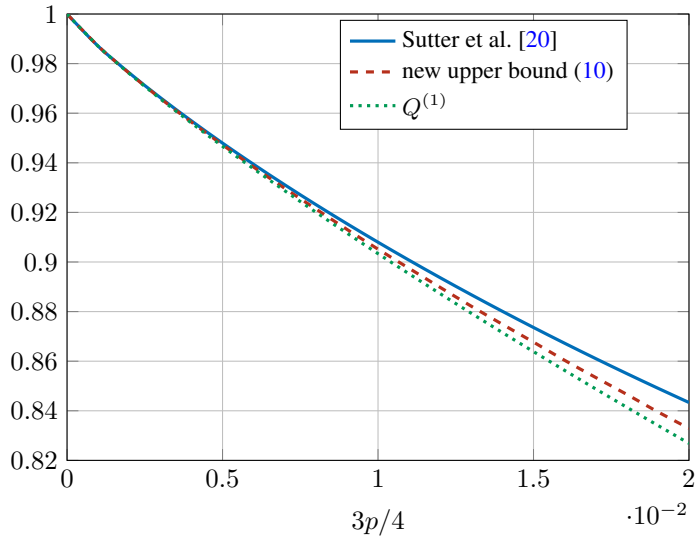 
%\begin{figure}[H]
%	\centering
%	\includegraphics[width=8.8cm]{depo_low_noise}
%	\caption{Upper and lower bounds on the quantum capacity $Q(\cD_p)$ of the low-noise qubit depolarizing channel for the interval $3p/4 \in\left[0,0.02\right]$. The channel's coherent information $Q^{(1)}$ (green, dotted) provides a lower bound on the quantum capacity $Q(\cD_p)$. The red dashed line depicts our bound $Q_{pf}$ in Proposition~\ref{prop:pure state bound}, which improves the bound via approximate degradable channels obtained in~\cite{Sutter2014}.}
%	\label{fig:depo low noise}
%\end{figure}

%\begin{figure}[H]
%	\centering
%	\includegraphics[width=8.8cm]{depo_middle_noise}
%	\caption{Upper and lower bounds on the quantum capacity $Q(\cD_p)$ of the intermediate-noise qubit depolarizing channel for the interval $3p/4\in\left[0,0.15\right]$. The channel's coherent information $Q^{(1)}$ (green, dotted) provides a lower bound on the quantum capacity $Q(\cD_p)$. 
%		The red dashed line depicts our bound $Q_{pf}$ in Proposition~\ref{prop:pure state bound}, which improves the bound obtained via mixed flag states~\cite{Fanizza2019} (blue, solid)  as well as the bound obtained in \cite{Leditzky2017} (yellow, solid).
%	}
%	\label{fig:depo middle noise}
%\end{figure}

\begin{figure}[H]
\centering
  \begin{tikzpicture}
	\begin{axis}[
		height=7.5cm,
		width=9.8cm,
		grid=major,
		xlabel=$3p/4$,
%		ylabel=$\beta$,
		xmin=0,
		xmax=0.15,
		ymax=1.0,
		ymin=0.1,
	    xtick={0,0.05,0.1,0.15},
        ytick={1,0.9,0.8,0.7,0.6,0.5,0.4,0.3,0.2,0.1},
		legend style={at={(0.68,0.995)},anchor=north,legend cell align=left,font=\footnotesize} 
	]
	
%%%%%%%%%%%%%%%%%%%%%%%%%%%%%%%%%%
% fillings here to not confuse the legend numbering (HAVE TO GO FIRST --> UNDERLAY)
%%%%%%%%%%%%%%%%%%%%%%%%%%%%%%%%%%
%%%%%%%%%%%%%%%%%%%%%%%%%%%%%%%%%%
% actual curves start here!!
%%%%%%%%%%%%%%%%%%%%%%%%%%%%%%%%%%	
%%%%%%%%
		\addplot[YellowOrange,very thick,smooth] coordinates {
 (0,	1)
(0.005,	0.954556437)
(0.01,	0.91908048)
(0.015,	0.887334439)
(0.02,	0.857984146)
(0.025,	0.830390602)
(0.03,	0.804173284)
(0.035,	0.779077409)
(0.04,	0.754919225)
(0.045,	0.73155926)
(0.05,	0.708887719)
(0.055,	0.686815853)
(0.06,	0.66527052)
(0.065,	0.644190596)
(0.07,	0.62352447)
(0.075,	0.603228248)
(0.08,	0.583264389)
(0.085,	0.563600651)
(0.09,	0.54420925)
(0.095,	0.525066174)
(0.1,	0.506150613)
(0.105,	0.487444487)
(0.11,	0.468932047)
(0.115,	0.450599544)
(0.12,	0.432434939)
(0.125,	0.41442767)
(0.13,	0.396568447)
(0.135,	0.378849075)
(0.14,	0.361262318)
(0.145,	0.343801763)
(0.15,	0.326461723)

	};
		\addlegendentry{Leditzky et al.~\cite{Leditzky2017}}

		\addplot[NavyBlue,very thick,smooth] coordinates {
 
(0,	1)
(0.005,	0.95167862)
(0.01,	0.913430104)
(0.015,	0.879029634)
(0.02,	0.847154655)
(0.025,	0.817177459)
(0.03,	0.788728665)
(0.035,	0.761564487)
(0.04,	0.735511905)
(0.045,	0.710441779)
(0.05,	0.686254078)
(0.055,	0.662869083)
(0.06,	0.640221817)
(0.065,	0.618258345)
(0.07,	0.596933219)
(0.075,	0.57620766)
(0.08,	0.556048223)
(0.085,	0.536425798)
(0.09,	0.517314852)
(0.095,	0.498692832)
(0.1,	0.480539702)
(0.105,	0.462837567)
(0.11,	0.445570373)
(0.115,	0.428723661)
(0.12,	0.41228436)
(0.125,	0.396240625)
(0.13,	0.380581689)
(0.135,	0.365297744)
(0.14,	0.350379845)
(0.145,	0.335819814)
(0.15,	0.321610175)

	};
		\addlegendentry{Fannizza et al.~\cite{Fanizza2019}}	
		
			%old UB
		\addplot[dashed,BrickRed,very thick,smooth] coordinates {
 (0,	1)
(0.005,	0.947254792)
(0.01,	0.905173372)
(0.015,	0.867691901)
(0.02,	0.83293071)
(0.025,	0.800808926)
(0.03,	0.770116996)
(0.035,	0.741494282)
(0.04,	0.714285493)
(0.045,	0.687701625)
(0.05,	0.662820022)
(0.055,	0.638948594)
(0.06,	0.615997072)
(0.065,	0.593885673)
(0.07,	0.572548872)
(0.075,	0.551928563)
(0.08,	0.531973692)
(0.085,	0.512636524)
(0.09,	0.494740668)
(0.095,	0.47654237)
(0.1,	0.458849542)
(0.105,	0.442541949)
(0.11,	0.425788193)
(0.115,	0.410389579)
(0.12,	0.395408966)
(0.125,	0.379874674)
(0.13,	0.365654592)
(0.135,	0.351787403)
(0.14,	0.33825308)
(0.145,	0.325033465)
(0.15,	0.312110204)
	};
		\addlegendentry{new upper bound~\eqref{eq:depo bound}}

%%%%%%%%%%%%%%%%%%%%%%%%%%%%%%%%%%%%%%%%%%%		
	%lower
		\addplot[dotted,ForestGreen,very thick,smooth] coordinates {
(0,	1)
(0.005,	0.946660495)
(0.01,	0.903357239)
(0.015,	0.863864852)
(0.02,	0.826860207)
(0.025,	0.791715006)
(0.03,	0.758059267)
(0.035,	0.725648586)
(0.04,	0.694309311)
(0.045,	0.663911654)
(0.05,	0.634354918)
(0.055,	0.605558703)
(0.06,	0.577457331)
(0.065,	0.549996149)
(0.07,	0.523128974)
(0.075,	0.496816268)
(0.08,	0.47102381)
(0.085,	0.445721691)
(0.09,	0.420883558)
(0.095,	0.396486014)
(0.1,	0.372508156)
(0.105,	0.348931198)
(0.11,	0.325738167)
(0.115,	0.302913659)
(0.12,	0.280443635)
(0.125,	0.258315244)
(0.13,	0.23651669)
(0.135,	0.215037102)
(0.14,	0.193866438)
(0.145,	0.172995392)
(0.15,	0.15241532)

	};
	\addlegendentry{$Q^{(1)}$}
	
%%%%%%%%%%%%%%%%%%%%%%%%%%%%%%%%%%%%%%%%%%%	

%%%%%%%%%%%%%%%%%%%%%%%%%%%%%%%%%%%%%%%%%%%		
	\end{axis}  
%%%%%%%%%%%%%%%%%%%%%%%%%%%%%%%%%%
% labeling here
%%%%%%%%%%%%%%%%%%%%%%%%%%%%%%%%%%

\end{tikzpicture}
\caption{Upper and lower bounds on the quantum capacity $Q(\cD_p)$ of the intermediate-noise qubit depolarizing channel for the interval $3p/4\in\left[0,0.15\right]$. The channel's coherent information $Q^{(1)}$ (green, dotted) provides a lower bound on the quantum capacity $Q(\cD_p)$. 
		The red dashed line depicts our bound in Eq.~\eqref{eq:depo bound}, which improves the bound obtained via mixed flag states~\cite{Fanizza2019} (blue, solid)  as well as the bound obtained in \cite{Leditzky2017} (yellow, solid).
		}
\label{fig:depo middle noise}
\end{figure}
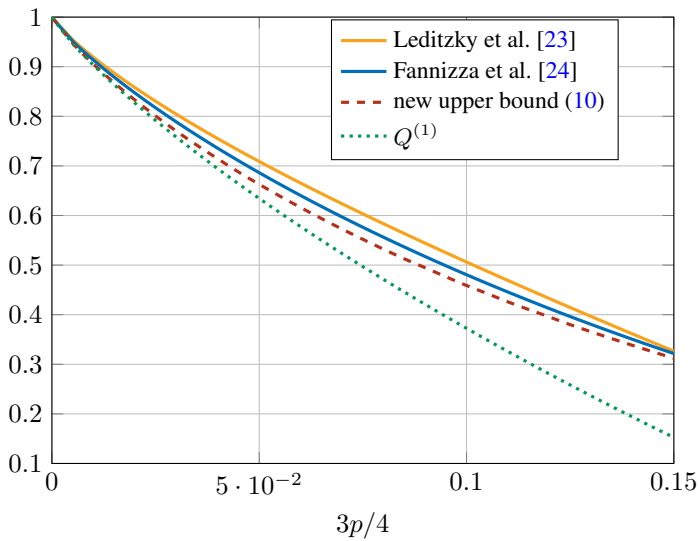

\subsection{Upper bounds on the quantum capacity of a general channel}
%For a general quantum channel, we further introduce an approach to upper bound the quantum capacity by optimizing the flagged (or extended) quantum channel, whose capacity is no smaller than the main channel since the flags (or extensions) of the channel can be considered as the assistance from the environment. Our method allows approximate degradability of the flagged channel and focuses on optimizing the quantum capacity of the approximate degradable flagged quantum channels with parameterized flag states. %Built on previously works~\cite{Sutter2014,Fanizza2019}, our method, on one hand, improves the approximate degradable bound by allowing more freedom or space in optimizing or searching the approximate degradable channels. On the other hand, our method improves the previous quantum flag upper bound by involving the optimization over the flagged channels. 

Degradable quantum channels are among the only channels whose quantum capacity is known \cite{Devetak2005}. A channel $\cN_{A\to B}$ is called degradable if there exists another quantum channel  $\cD_{B \to E}$ such that 
$\cN_{A\to E}^{c}= \cD_{B \to E} \circ \cN_{A\to B}$, 
where the complementary channel $\cN_{A\to E}^{c}(\rho):= \tr_B V\rho V^\dagger$ and $V$ is an isometry that ensures $\cN_{A\to B}(\rho):= \tr_E V\rho V^\dagger$. Here, the channel $\cD_{B \to E}$ is called degrading channel.
For degradable quantum channels,  its quantum capacity is given by the channel's coherent information \cite{Devetak2005}. 
%To be specific, 
%On the other hand, the channel $\cN_{A\to B}$ is called anti-degradable if there exists a quantum channel $D_{E\to B}$ such that 
%\begin{align}
%	\cN_{A\to B} = D_{E \to B} \circ \cN_{A\to E}^c.
%\end{align}

%The class of degradable channels provides us a chance to explore the quantum capacity of other channels since its quantum capacity is given by a single letter formula. 
In a previous work, Sutter et al.~\cite{Sutter2014} made use of the continuity bound of channel capacities~\cite{Leung2009} to estimate the quantum capacity of approximate degradable channels and gave previously tightest upper bound for low noise depolarizing channels. 
%Moreover, degradable extensions with orthogonal flags have also been used to provide estimations on the quantum capacity of depolarizing channels~\cite{Smith2008b,Ouyang2014,Leditzky2017}. 
By decomposing the depolarizing channel to the convex combination of degradable and anti-degradable channels, Leditzky et al.~\cite{Leditzky2017} established the state-of-the-art upper bound in the high-noise regime. Recently, Fanizza et al.~\cite{Fanizza2019} showed that the degradable extensions with non-orthogonal quantum flags could be used to obtain better upper bound on the quantum capacity of the intermediate-noise qubit depolarizing channel. 

The intuition behind the \textit{quantum flags}~\cite{Fanizza2019} can be traced back to environment-assisted quantum information processing~\cite{Gregoratti2003,Smolin2005,Winter2005}, where an environment is assumed to ``friendly'' provide partial helpful information to the receiver. To be specific, for a given noisy quantum channel $\cN_{A\to B}=\sum_j p_j\cN_j$ with probability distribution $\{p_j\}$ and individual CPTP maps $\cN_j$, let us assume that the state of the partial environment is $\sigma_j$ for each $\cN_j$ and the receiver has access to this partial environment. In this case, the complete channel can be written as	$\widehat\cN =\sum_j p_j\cN_j \otimes \sigma_j$. With such assistance, the receiver Bob receives a quantum flag $\sigma_j$ which encodes the information about which channel is acting. Thus the flagged channel can certainly convey more quantum information than the original channel $\cN$, i.e.,
\begin{align} \label{eq:flag bound}
	Q(\cN) \le Q(\widehat\cN).
\end{align}
Note that a special kind of flagged channel with orthogonal flags was considered in \cite{Smith2008b,Ouyang2014}. However, the flag states are not necessarily orthogonal.

%For the depolarizing channel, one of its flagged versions can be represented in the following form
%$\widehat\cD_{p}^d(\rho) = (1-p)\rho\otimes\sigma_0 + p\tr(\rho)\frac{I_d}{d}\otimes\sigma_1$. 
%Fanizza et al. \cite{Fanizza2019} studied the coherent information of the degradable flagged depolarizing channel $\widehat\cD_{p}^d$. 
%By choosing $\sigma_0=\proj{e_1}$, $\sigma_1=c\proj{e_1}+(1-c)\proj{e_1^\perp}$ and constructing a special kind of degrading channel,
%the authors of \cite{Fanizza2019} derived sufficient conditions for the degradability of $\widehat\cD_{p}^d$ and determine a feasible region of $c$, which leads to an outer region of $Q(\widehat\cD_{p}^d)$. However, the step of manually creating a degrading channel leaves the space for finding the optimal flagged degradable channel with minimal coherent information. Moreover, as the step of constructing the degrading channel is artificial, this method can not be easily applied to other quantum channels.

To better exploit the power of the extended or flagged quantum channel, we introduce a more general and efficient approach that optimizes the flagged states and utilizes approximate degradable quantum channels. The central idea of our approach is to consider the optimization of all degrading channels with some error tolerance rather than manually constructing the degrading channel.  We also introduce a more general quantum flag model involving the decomposition of a quantum channel to the sum of a set of CP maps.

Given a quantum channel $\cN_{A\to B}$, its degradability parameter \cite{Sutter2014} is defined to be
\begin{align}
\eta(\cN_{A\to B}) :=  \min_{\cD_{B \to E}}\frac{1}{2} \|\cN_{A\to E}^c - \cD_{B \to E}\circ \cN\|_\diamond,
\end{align}
where  $\|\cF\|_\di:=\sup_{k\in \mathbb N} \sup_{\|X\|_1 \leq 1}\|(\cF\ox \id_k)(X)\|_1$ denotes the diamond norm. Moreover, a quantum channel $\cN$ is called $\ve$-degradable if $\eta(\cN)=\ve$, and $\eta(\cN)$ can be efficiently computed via semidefinite programming~\cite{Sutter2014}. More details are provided in the Appendix~\ref{sec:app A}.

For approximate degradable channels, a useful property related to quantum capacity was established in \cite{Sutter2014}: let $\cN_{A\to B}$ be a $\ve$-degradable quantum channel whose environment dimension is $d_E$, then
\begin{align}\label{eq:approx bound}
	Q(\cN) \le Q^{(1)}(\cN) + {\varepsilon} \log (d_E-1)+h\left({\varepsilon}\right)+ 2\varepsilon \log  d_E+g(\ve),
\end{align}
where $g\left(p\right)= (1+p)h({p}/{(1+p)})$ is the bosonic entropy and $h\left(p\right)=-p\log p - (1-p)\log(1-p)$
 is the binary entropy.
Using this property, we further introduce the approximate degradable flag channels and obtain the following upper bound on quantum capacity. Here we use CP map decomposition of the channel instead of the decomposition of CPTP maps since the latter is more restricted but not necessary in this task. 
\begin{proposition}[General flag approximate degradable bound]\label{prop:q general}
Let $\cN$ be a quantum channel with CP map decomposition $\cN=\sum_{j=0}^k\cN_j$. The quantum capacity of $\cN$ is bounded by
\begin{align}
	Q(\cN) \le \inf_{\sigma_0, \cdots, \sigma_k} Q^{(1)}(\widehat\cN) + {\eta(\widehat\cN)} \log (d_E-1)+h\left({\eta(\widehat\cN)}\right)+2\eta(\widehat\cN) \log  d_E+  g\left(\eta(\widehat\cN)\right).
\end{align}
where $\widehat\cN (\cdot)= \sum_{j=0}^k \cN_j(\cdot) \otimes \sigma_j$ and $d_E$ is the dimension of  the output system of $\widehat\cN^c$ that depends on the number of the flags and the dimension of the flag states.  %Note that this bound is never worse than the previous continuity bound~\cite{Sutter2014} since the later is the case $\sigma_0=\sigma_1$.
\end{proposition}
\begin{proof}
The key idea here is to consider the  approximate degradable flag channels. By evaluating the degradability parameter of the flag channel and applying the continuity upper bound on the quantum capacity, we have
\begin{align}
	Q(\cN) &\le \inf_{\sigma_0,\cdots \sigma_k} Q(\widehat \cN) \label{eq:relax to flag} \\
		   &\le \inf_{\sigma_0, \cdots,\sigma_k} Q^{(1)}(\widehat \cN) + {\varepsilon} \log (d_E-1)+h\left({\varepsilon}\right)+2\varepsilon \log  d_E+g\left({\varepsilon}\right),\label{eq:relax approx bound}
\end{align}
where $\ve$ in Eq.~\eqref{eq:relax approx bound} is the degradability parameter determined by $\eta(\widehat\cN)$. The first inequality is due to the fact that $\widehat\cN$ is an extension of the original channel $\cN$, which allows Bob to obtain extra information from the environment. The second inequality is due to Eq.~\eqref{eq:approx bound}.
\end{proof}

Computing the above bound may be extremely hard. However, we could restrict the number of flags and parametrize the flag states and then optimize these parameters. In particular, for the case of two pure flag states, we are able to establish the following upper bound. As the system of the flag states is at the receiver Bob's side, then he could locally do any unitary on the system. Thus, without loss of generality, we could fix one flag state $\sigma_1 = \proj 0$ and assume that $\sigma_0 =\proj {\psi_a}$.
\begin{proposition}[Pure flag approximate degradable bound]\label{prop:pure state bound}
Let $\cN$ be a quantum channel with CP map decomposition $\cN=\cN_0+\cN_1$. The quantum capacity of $\cN$ is bounded by the following optimized pure flag upper bound:
	\begin{align}\label{eq:QPF bound}
	 Q_{pf}(\cN):= \inf_{0\le \alpha\le 1} f(\cN_0,\cN_1,\alpha),
	 \end{align}
	 with
	 \begin{align*}
     f(\cN_0,\cN_1,\alpha):= Q^{(1)}(\widehat\cN) + {\eta(\widehat\cN)} \log (d_E-1)+h\left({\eta(\widehat\cN)}\right)+2\eta(\widehat\cN) \log  d_E
     + g\left(\eta(\widehat\cN)\right),
	\end{align*}
where $\widehat\cN (\cdot)=\cN_0(\cdot) \otimes \proj{\psi_\a} + \cN_1(\cdot) \otimes \proj 0$, $\ket{\psi_\a} = \sqrt \alpha \ket 0+\sqrt{1-\alpha}\ket1$, and $d_E$ is the dimension of  the output system of $\widehat\cN^c$ that depends on the flag states. 
%Note that this bound is never worse than the previous continuity bound~\cite{Sutter2014} since the later is just the case $\alpha=1$.
\end{proposition}

Note that the objective function $f(\cN_0,\cN_1,\alpha)$ in Proposition~\ref{prop:pure state bound} can be computed efficiently if the  coherent information of $\widehat{\cN}$ can be computed. Then we could do brute-force search on $\a\in[0,1]$ to find the almost optimal flag states in this case and obtain the almost optimal pure flag upper bound on quantum capacity. Similarly, for the case $\sigma_0=\proj{e_1}$, $\sigma_1=c\proj{e_1}+(1-c)\proj{e_1^\perp}$ considered in \cite{Fanizza2019}, we could also use the brute-force search approach to find the almost optimal quantum flag upper bound.	We also note that the above bounds may be slightly improved by using a slightly tighter continuity bound in Theorem 7 of \cite{Sutter2014}. 

Built on previously works~\cite{Sutter2014,Fanizza2019}, our bound in Eq.~\eqref{eq:QPF bound}, on one hand, improves the approximate degradable bound by allowing more freedom or space in optimizing or searching the approximate degradable channels. On the other hand, our method improves the previous quantum flag upper bound by involving the optimization over the flagged channels. 

Furthermore, we observe that one can only optimize over the flagged channels which are degradable and obtain the following upper bound on quantum capacity. This bound can be efficiently computed since the coherent information of degradable channels can be efficiently optimized~\cite{Fawzi2017a}. Note that we could also use the extendibility measures in \cite{WWW19} to replace the condition $\eta(\widehat \cN)=0$ due to the fact that the main channel is degradable if and only if its complementary channel is anti-degradable (or two-extendible).
\begin{corollary}[Pure flag degradable bound]\label{prop:Q bound special}
	Let $\cN$ be a quantum channel with CP map decomposition $\cN=\cN_0+\cN_1$. The   quantum capacity of $\cN$ satisfies that
	\begin{align}\label{eq:Q pf deg}
	Q(\cN)\le	\widetilde Q_{pf}(\cN):= \inf_{0\le \alpha\le 1} \left\{ Q^{(1)}(\widehat\cN)\big| \eta(\widehat\cN) = 0 \right\},
	\end{align}
	where $\widehat\cN (\cdot) = \cN_0(\cdot) \otimes \proj{\psi_\a} + \cN_1(\cdot) \otimes \proj 0$ and  $\ket{\psi_\a} = \sqrt \alpha \ket 0+\sqrt{1-\alpha}\ket1$. If there is no $\widehat{\cN}$ such that $\eta(\widehat\cN) = 0$, the value of this bound is set to be infinity.
\end{corollary}
\begin{remark}\label{remark 1}
We note that this Corollary could be the more useful than the above proposition since it is generally not clear how to compute the coherent information of an arbitrary flagged channel. Moreover,
it is worth noting that a useful variant of Corollary \ref{prop:Q bound special} is to choose degradable decompositions and orthogonal flag states, then the flagged channel will be degradable and thus leads to efficiently computable bounds.
\end{remark}
%\begin{proof}
%For the proof,	the key idea is to restrict the degradability parameter $\eta(\widehat\cN)$ in the objective function $f(\cN_0,\cN_1,\alpha)$ in \eqref{eq:QPF bound} to be $0$, which leads to $f(\cN_0,\cN_1,0) = Q^{(1)}(\widehat\cN)$. Hence, the whole optimization problem is to minimize the coherent information of the flag channel with zero degradability parameter.	Note that we could also use the extendibility measures in \cite{WWW19} to replace the condition $\eta(\widehat \cN)=0$ due to the fact that the main channel is degradable if and only if its complementary channel is anti-degradable (or two-extendible).
%\end{proof}

\subsection{Improved upper bounds for the GAD channel}
We further apply our bounds to other quantum channels of theoretical and experimental interest. We establish improved upper bounds on the quantum capacity of the generalized amplitude damping channel, which are tighter than previously known upper bounds~\cite{Khatri2019} via data-processing approach~\cite{Khatri2019} and Rains information~\cite{Tomamichel2015a}.

The generalized amplitude damping (GAD) channel is one of the realistic sources of noise in superconducting quantum processor~\cite{Chirolli2018}. It can be viewed as the qubit analogue of the bosonic thermal channel and can be used to model lossy processes with background noise for low-temperature systems. When $N = 0$, it reduces to the conventional amplitude damping channel. For the GAD channel, some prior works have established bounds on its quantum capacity~\cite{Rosati2018a,Khatri2019,Patron2009}.
and improved lower bounds on the quantum capacity of the GAD channel were derived in \cite{Bausch2018}.

To be specific, the generalized amplitude damping channel is a two-parameter family of channels described as follows:
\begin{equation}
\mathcal{A}_{y, N}(\rho)=A_{1} \rho A_{1}^{\dagger}+A_{2} \rho A_{2}^{\dagger}+A_{3} \rho A_{3}^{\dagger}+A_{4} \rho A_{4}^{\dagger},
\end{equation}
where $y,N\in[0,1]$ and
\begin{align}
A_{1}&=\sqrt{1-N}(|0\rangle\langle 0|+\sqrt{1-y}| 1\rangle\langle 1|),\\
A_{2}&=\sqrt{y(1-N)}|0\rangle\langle 1|,\\
A_{3}&=\sqrt{N}(\sqrt{1-y}|0\rangle\langle 0|+| 1\rangle\langle 1|),\\
A_{4}&=\sqrt{y N}|1\rangle\langle 0|.
\end{align}

To apply our bounds, we use the technique mentioned in Remark \ref{remark 1}. We first decompose the GAD channel as the following convex combination of two degradable channels:
\begin{align}
\cA_{y,N} = \sqrt{1-N} \cA^1_{y,N}(\rho) + \sqrt{N} \cA^2_{y,N}(\rho),
\end{align}
with $\cA^1_{y,N}(\rho) = A_1 \rho A_1^\dagger/(1-N)+A_2 \rho A_2^\dagger/(1-N)$ and  $\cA^2_{y,N}(\rho) = A_3 \rho A_3^\dagger/N+A_4 \rho A_4^\dagger/N$.
Then, we could apply the pure flag approximate degradable bound in Proposition~\ref{prop:pure state bound} by introducing the flagged channel
\begin{align}
\widehat{\cA}_{y,N} = \sqrt{1-N} \cA^1_{y,N}(\rho) \otimes \proj{\psi_\a} + \sqrt{N} \cA^2_{y,N}(\rho)\otimes \proj 0,
\end{align} 
where $\ket{\psi_\a} = \sqrt{\a}\ket0+\sqrt{1-\a}\ket 1$.

By numerical optimization, we observe that the optimal value is achieved by $\a=0$ for noise parameter $y\in(0,1/2)$. By further exploration, we find that the flagged channel $\widehat{\cA}_{y,N}$ with $\ket{\psi_a}=\ket 1$ is degradable due to the fact that both $\cA^1_{y,N}$ and $\cA^2_{y,N}$ are degradable (cf.~Lemma 4 of \cite{Smith2008b}).

As shown in Appendix~\ref{sec:GAD bound}, we derive the following bound on the quantum capacity of GAD channel.  
\begin{proposition}
For the GAD channel with noise parameter $y\in (0,1/2)$ and $N\in(0,1)$, it holds that
\begin{align}
Q({\cA}_{y,N}) &\le  Q(\widehat{\cA}_{y,N})= \max _{p \in[0,1]} I_{\mathrm{c}}\left(p\proj0+(1-p)\proj1, \widehat{\mathcal{A}}_{\gamma, N}\right), \label{eq:GAD bound}
\end{align}
where $I_{\mathrm{c}}\left(\rho, \cN\right) \equiv H(\mathcal{N}(\rho))-H\left(\mathcal{N}^{c}(\rho)\right)$ and $\widehat{\cA}_{y,N} = \sqrt{1-N} \cA^1_{y,N}(\rho) \otimes \proj{1} + \sqrt{N} \cA^2_{y,N}(\rho)\otimes \proj 0$.
\end{proposition}
%\begin{proof}
%	Since $\widehat{\cA}_{y,N}$ is the flagged channel of  ${\cA}_{y,N}$, we have 
%	 \begin{align}
%	 Q({\cA}_{y,N}) \le  Q(\widehat{\cA}_{y,N}).
%	 \end{align}
%	 Then, we have
%	 \begin{align}
%	 Q({\cA}_{y,N}) &\le   Q(\widehat{\cA}_{y,N})\\
%	 & = Q^{(1)}(\widehat{\cA}_{y,N})\\
%	 &= \max _{p \in[0,1]} I_{\mathrm{c}}\left(p\proj0+(1-p)\proj1, \widehat{\mathcal{A}}_{\gamma, N}\right). 
%	 \end{align}
%	 The first equality holds since $\widehat{\cA}_{y,N}$ is degradable. The second inequality is due to fact that the coherent information for degradable channels is a concave function, which implies that diagonal input states outperform non-diagonal states~\cite{Wolf2007}.
%\end{proof}

In Figure~\ref{fig:GAD}, we compare our bound with the previously best known bounds obtained in \cite{Khatri2019} on the quantum capacity of the GAD channel. We mainly compare our bound with the data-processing bound~\cite{Khatri2019} defined as
\begin{align}
	Q(\mathcal{A}_{\gamma, N})  \leq Q_{DP}(\mathcal{A}_{\gamma, N}) \equiv Q(\mathcal{A}_{{y(1-N)}/{(1-yN)}, 0})
\end{align}
and
 the Rains information~\cite{Tomamichel2015a} defined as follows:
\begin{align}
R(\cN_{A'\to B})  \equiv \max _{\rho_{A} \in \mathcal{S}(A)} \min _{\sigma_{AB} \in \operatorname{PPT}^{\prime}(A: B)} {D}\left(\mathcal{N}_{A^{\prime} \rightarrow B}\left(\phi_{A A^{\prime}}\right) \| \sigma_{A B}\right),	
\end{align}
where $\phi_{AA'}$ is a purification of $\rho_A$, $\operatorname{PPT}^{\prime}(A: B) \equiv\{\sigma_{A B} | \sigma_{A B} \geq 0,\|\sigma_{A B}^{T_{B}}\|_{1} \leq 1\}$, $T_B$ is the partial transpose on system $B$, and $D\left(\rho \| \sigma\right)\equiv \tr\left[\rho\left(\log _{2} \rho-\log _{2} \sigma\right)\right]$ is the quantum relative entropy. Rains information can be seen as the channel version of the Rains bound~\cite{Rains2001}, whose variants have many applications in quantum communication. For readers interested in the Rains bound, we refer to \cite{Rains2001,Wang2016c,Audenaert2002,Wang2016,Bauml2019,Bauml2018} for more information. Codes for computing our bound for the GAD channel can be found in \cite{Wang_codes}. The codes for computing the coherent information and the Rains information of the GAD channel are from~\cite{Khatri2019}, where the authors utilized the $\sigma_Z$ covariance and concavity of Rains information in the input state to simplify the optimization.
Interestingly, we numerically find that our new bound (red-dashed line) achieves $0$ at $y=1/2$, which is in contrast to the previous bounds from \cite{Khatri2019}. These cases could be verified as anti-degradable channels utilizing the max-unextendible entanglement in \cite{WWW19}.

\begin{figure}[H]
	\centering
	\begin{adjustwidth}{0.2cm}{0.2cm}
		\begin{tikzpicture}
		% \draw (-8,-2) to[grid with coordinates] (8,2);
		\node at (-5.35,0.38) {\includegraphics[width = 5.68cm]{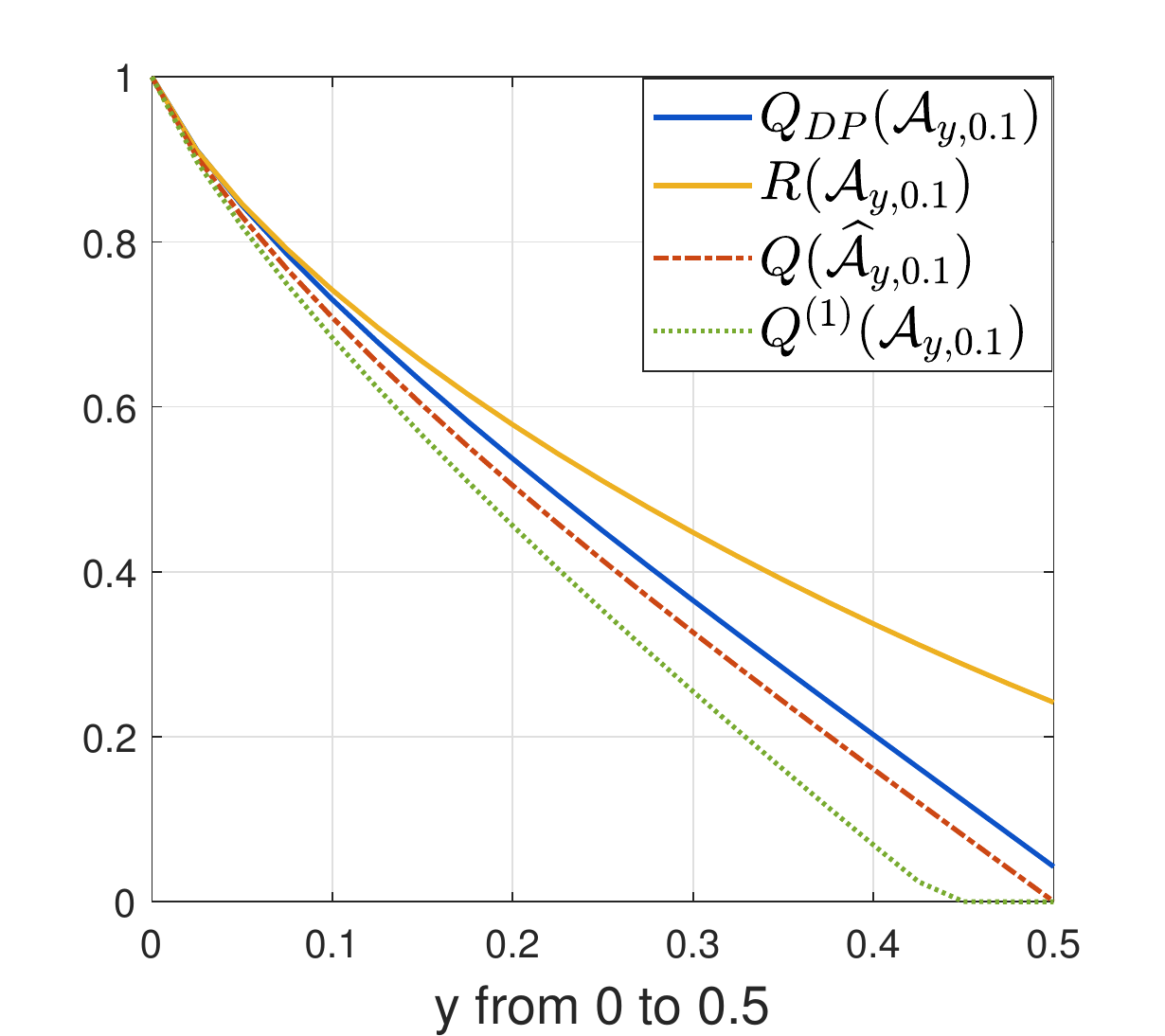}};
		\node at (-5.35,-2.6) {\small (a)   $\cA_{y,0.1}$};
		
		\node at (0,0.38) {\includegraphics[width = 5.68cm]{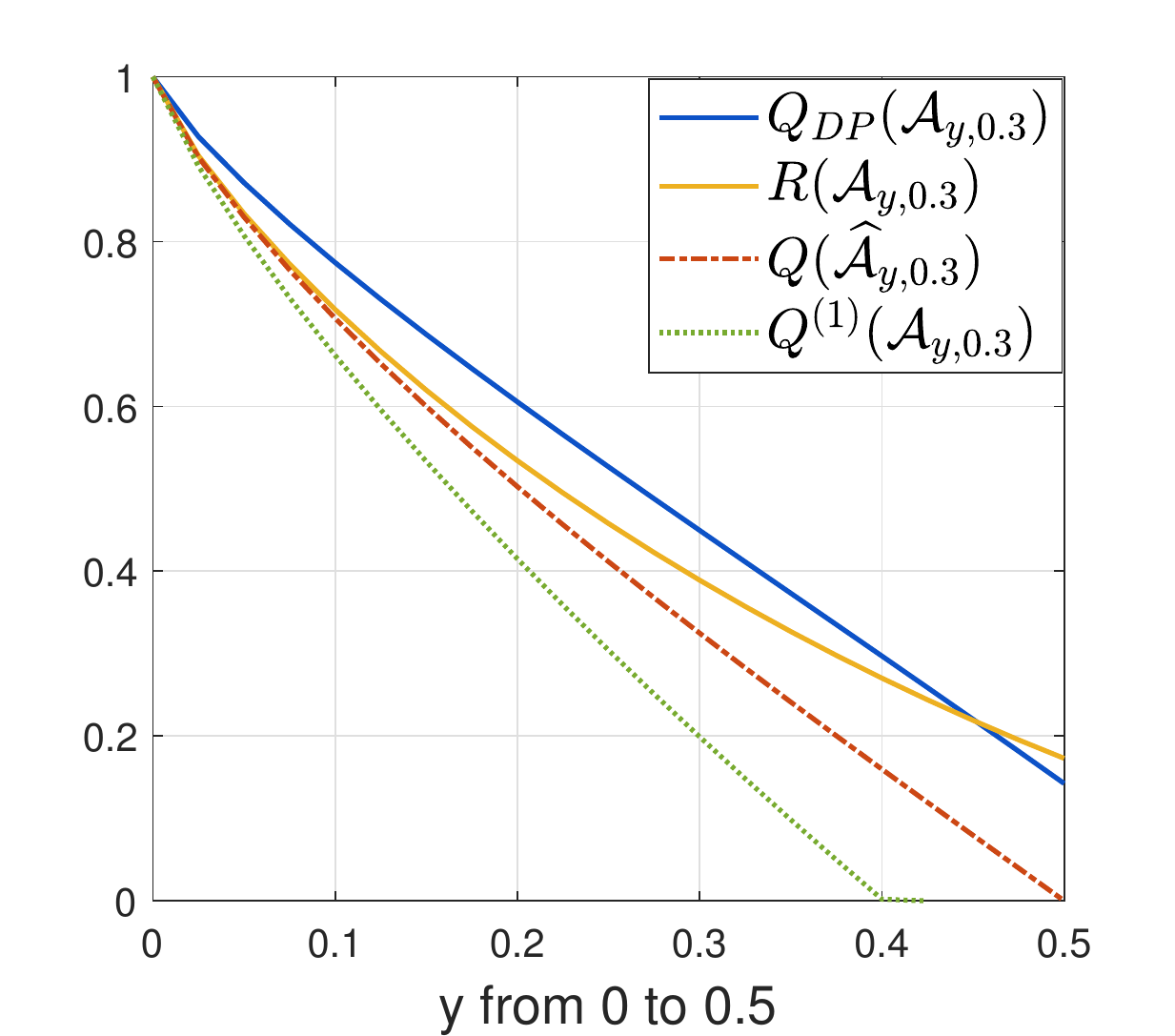}};
		\node at (0,-2.6) {\small (b) $\cA_{y,0.3}$};
		
		\node at (5.35,0.38) {\includegraphics[width = 5.68cm]{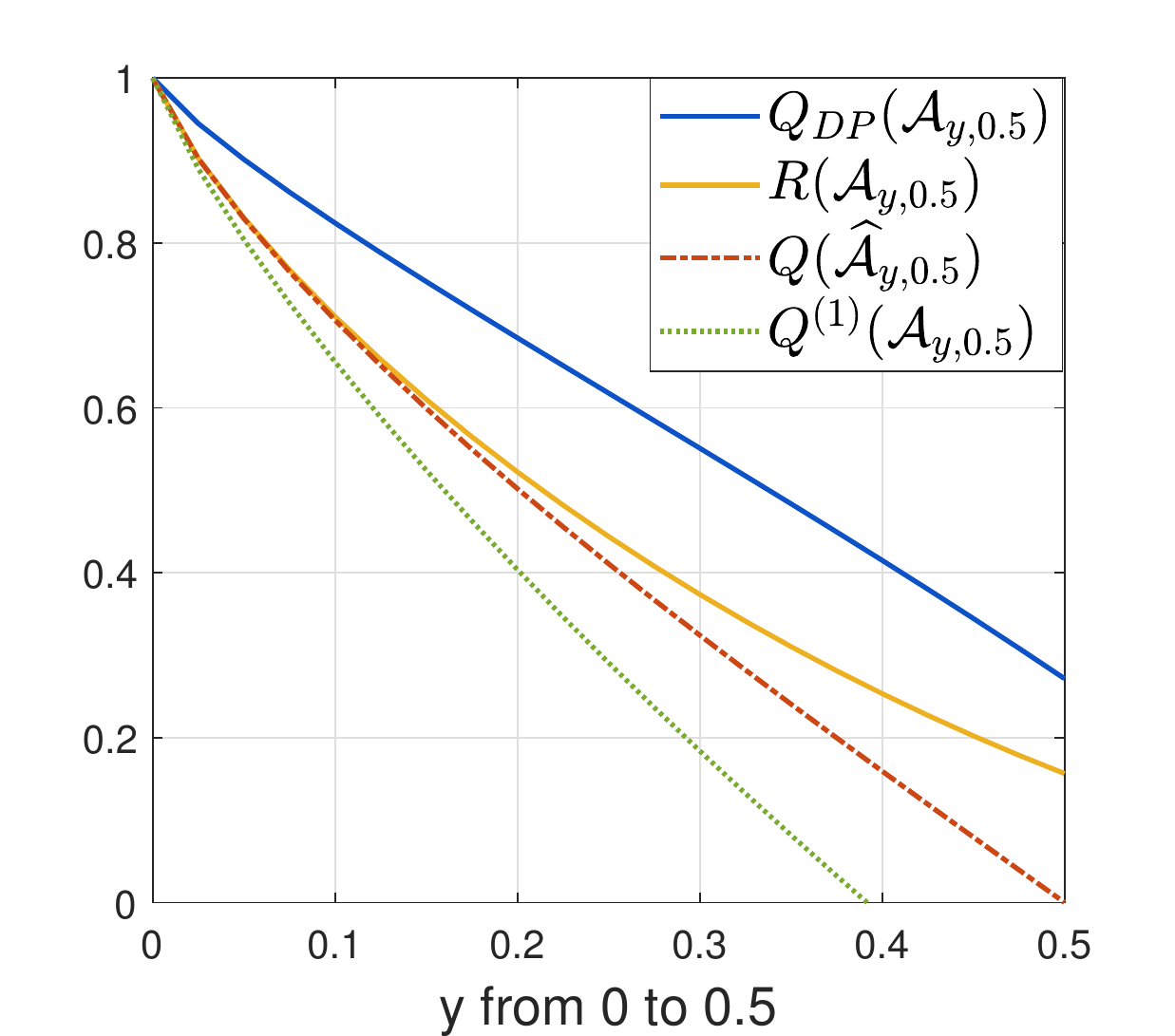}};
		\node at (5.35,-2.6) {\small (c)  $\cA_{y,0.5}$};
		\end{tikzpicture}
	\end{adjustwidth}
	\caption{\small  Upper and lower bounds on the quantum capacity $Q(\cA_{y,N})$ of the generalized amplitude damping (GAD) channel. The green dotted line depicts the lower bound. The red dashed line depicts our bound in Eq.~\eqref{eq:GAD bound}, which improves the bounds via data-processing approach and Rains information obtained in~\cite{Khatri2019}.}
	\label{fig:GAD}
\end{figure}
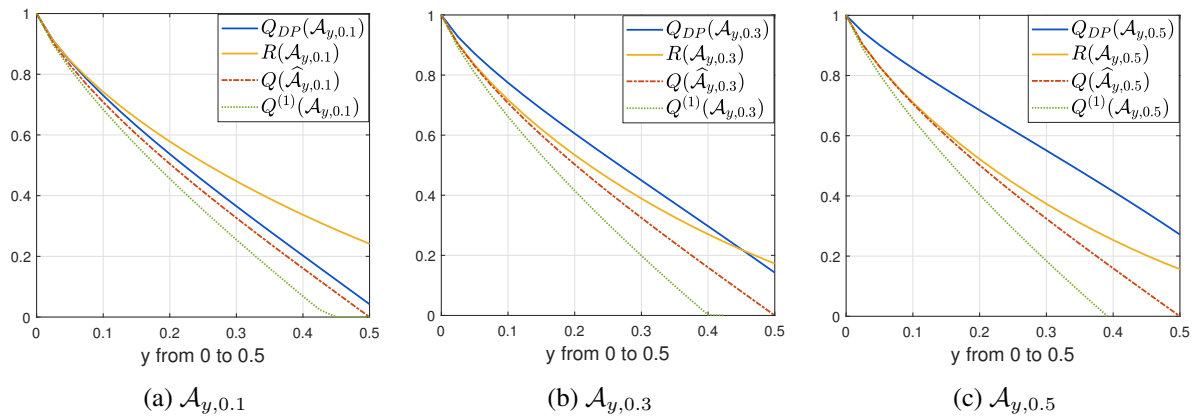

%\section{Extension to private communication}\label{sec:P bound}
\subsection{Fundamental limits on private communication}
The private capacity of a quantum channel is defined as the optimal rate at which classical information can be transmitted faithfully and privately from the sender Alice to the receiver Bob. The privacy here means a third party Eve who has access to the channel's environment cannot obtain anything useful about the information that Alice sends to Bob.
The private capacity theorem states that the private capacity of a noisy quantum channel is given by its regularized private information~\cite{Cai2004,Devetak2005}
\begin{equation}
P(\mathcal{N})=\lim _{n \rightarrow \infty} \frac{1}{n} P^{(1)}\left(\mathcal{N}^{\otimes n}\right)=\sup _{n \in \mathbb{N}} \frac{1}{n} P^{(1)}\left(\mathcal{N}^{\otimes n}\right),
\end{equation}
where $P^{(1)}(\mathcal{N}) = \max _{{E}}\left[\chi({E}, \mathcal{N})-\chi\left({E}, \mathcal{N}^{c}\right)\right]$ is the private information with the maximization taken over all possible ensembles $E=\{p_j,\rho_j\}$ and $\chi({E}, \mathcal{N}) = H\left(\sum_{i} p_{i} \mathcal{N}\left(\rho_{i}\right)\right)-\sum_{i} p_{i} H\left(\mathcal{N}\left(\rho_{i}\right)\right)$ is the Holevo information of the ensemble. The private capacity of a quantum channel is also extremely difficult to solve and only some weak converse or strong converse bounds were established in the past two decades (see, e.g., \cite{Takeoka2014,Berta2013,Christandl2016,Pirandola2015b,Wilde2016c}).
In particular, for a degradable channel $\cN$, the private capacity coincides with the coherent information of the channel~\cite{Smith2007a}, which implies that $P(\cN)=Q(\cN)=Q^{(1)}(\cN)$.

Previously the upper bound on the private capacity of an $\varepsilon$-degradable channel $\cN$ was established as follows~\cite{Sutter2014,Leditzky2017a}:
\begin{align}
	  P(\mathcal{N}) \leq Q^{(1)}(\mathcal{N})+2\ve \log (d_E-1)+8 \ve \log d_E + 2 h\left({\ve}\right)+4 g\left(\ve\right),
\end{align}
where $\ve=\eta(\cN)$ and $d_E$ is the dimension of the environment.

Using similar ideas as in Proposition~\ref{prop:pure state bound} and Corollary~\ref{prop:Q bound special}, we obtain the following upper bounds on the private capacity of a quantum channel.
Let $\cN$ be a quantum channel with CP map decomposition $\cN=\cN_0+\cN_1$. The private capacity of $\cN$ is bounded by the following optimized pure flag upper bound:
	\begin{align}
	P_{pf}(\cN):= \inf_{0\le \alpha\le 1} \gamma(\cN_0,\cN_1,\alpha),
	\end{align}
	with $\gamma(\cN_0,\cN_1,\alpha):= Q^{(1)}(\widehat\cN) + 2{\eta(\widehat\cN)} \log (d_E-1) +8\eta(\widehat\cN) \log  d_E + 2h\left({\eta(\widehat\cN)}\right)
	+ 4g\left(\eta(\widehat\cN)\right)$,
	where $\widehat\cN (\cdot)=\cN_0(\cdot) \otimes \proj{\psi_\a} + \cN_1(\cdot) \otimes \proj 0$, $\ket{\psi_\a} = \sqrt \alpha \ket 0+\sqrt{1-\alpha}\ket1$, and $d_E$ is the dimension of  the output system of  $\widehat\cN^c$. 
%Here we remark that the upper bound in Proposition~\ref{prop:private bound} for private capacity may be slightly improved by using a slightly tighter continuity bound in~\cite{Sutter2014,Sharma2017}.
By further restricting the extended channel to be degradable, we obtain the following bound. Let $\cN$ be a quantum channel with CP map decomposition $\cN=\cN_0+\cN_1$. The private capacity of $\cN$ satisfies that
\begin{align}
P(\cN)\le	\widetilde P_{pf}(\cN) = \widetilde Q_{pf}(\cN)= \inf_{0\le \alpha\le 1} \left\{ Q^{(1)}(\widehat\cN)\big| \eta(\widehat\cN) = 0 \right\}, \label{eq:private bound}
	\end{align}
	where $\widehat\cN (\cdot)=\cN_0(\cdot) \otimes \proj{\psi_\a} + \cN_1(\cdot) \otimes \proj 0$ and $\ket{\psi_\a} = \sqrt \alpha \ket 0+\sqrt{1-\alpha}\ket1$.
Note that the set could be empty, in which case the bound above could be taken to infinity. However, it may be possible to generalize this bound by employing multiple flags.

%\subsection{Applications to depolarizing channels and BB84 channels}
%\begin{example}(Private capacity of low-noise depolarizing channel)
%In Figure~\ref{fig:PN depo low noise}, we compare our bound in and the above upper bound on the private capacity for the low-noise depolarizing channel.
%
%\begin{figure}[H]
% 	\centering
% 	\includegraphics[width=8.1cm]{depo_low_noise_PN005}
% 	\caption{Upper and lower bounds on the private capacity $P(\cD_p)$ of the low-noise qubit depolarizing channel for the interval $3p/4 \in\left[0,0.05\right]$. The channel's coherent information $Q^{(1)}$ (green, dotted) provides a lower bound on the private capacity $P(\cD_p)$. The red dashed line depicts our bound $P_{pf}$ in Proposition~\ref{prop:private bound}, which improves the bound via approximate degradable channels obtained in~\cite{Sutter2014}.}
% 	\label{fig:PN depo low noise}
%\end{figure}
% \end{example}
% 
 
We further show applications of our bounds in quantum key distribution. It is well known that Bennett-Brassard quantum key distribution~\cite{Bennett1984} is the most widely studied and practically applied quantum cryptography protocol. It is thus desirable to obtain efficiently computable bounds on the achievable key rate. Here we apply our bounds and evaluate the secret key rate of this protocol. To start the analysis, we first introduce the so-called BB84 channel: a qubit Pauli channel with independent bit-flip and phase-flip error probability where $p_X\in[0,1/2]$ denotes the bit flip and $p_Z \in [0, 1/2]$ the phase flip probability. Due to its relevance for the BB84 protocol, this channel is known as the BB84 channel in the literature~\cite{Smith2008b}. Further information about the BB84 channel and its operational justification can be found in \cite{Wilde2019}.

More formally the BB84 channel is a CPTP map from $\cL(A)$ to $\cL(B)$:
\begin{align}
\cB_{p_X,p_Z}(\rho) = (1-p_X-p_Z+p_Xp_Z)\rho+(p_X-p_Xp_Z)X\rho X+(p_Z-p_Zp_X)Z\rho Z+p_Xp_ZY\rho Y,
\end{align}
with $d_A = d_B = 2$. It is easy to verify that a Bell state maximizes the coherent information and then the channel coherent information of the BB84 channel is given by 
\begin{align}
Q ^{(1)} (\cB_{p_X,p_Z})= 1 - h(p_X)- h(p_Z).
\end{align}
In particular, we will focus on the case where $p_X = p_Z =p$. In this case, Smith and Smolin~\cite{Smith2008b} showed that
\begin{equation}
Q\left(\mathcal{B}_{p, p}\right) \leq h\left(\frac{1}{2}-2 p(1-p)\right)-h(2 p(1-p)).
\end{equation}

To apply our bound, we first introduce the flagged channel of the BB84 channel 
$
\widehat\cB_{p,p} (\rho)= (1-p)^2\rho \otimes\proj{\psi_\a} +(p-p^2)X\rho X \otimes \proj 0+(p-p^2)Z\rho Z\otimes \proj0+p^2Y\rho Y\otimes\proj0$,
where $\ket {\psi_\a}=\sqrt{\a}\ket0+\sqrt{1-\a}\ket1$. Then we could use brute-force search to optimize the function in Eq.~\eqref{eq:private bound} and we numerically find degradable flagged channels for each $p$, whose coherent information could be efficiently computed. In Figure.~\ref{fig:PN BB84 low noise}, we numerically implement our bound to the BB84 channel and compare our bound with previously best-known bounds.

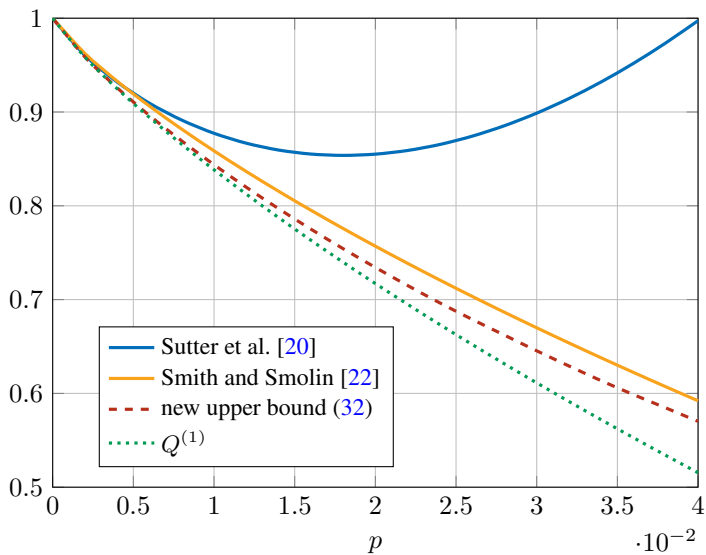
\begin{figure}[H]
\centering
  \begin{tikzpicture}
	\begin{axis}[
		height=7.8cm,
		width=10.11cm,
		grid=major,
		xlabel=$p$,
%		ylabel=$\beta$,
		xmin=0,
		xmax=0.04,
		ymax=1.0,
		ymin=0.5,
	    xtick={0,0.005,0.01,0.015,0.02,0.025,0.03,0.035,0.04},
        ytick={1,0.9,0.8,0.7,0.6,0.5},
		legend style={at={(0.3,0.343)},anchor=north,legend cell align=left,font=\footnotesize} 
	]
	
%%%%%%%%%%%%%%%%%%%%%%%%%%%%%%%%%%
% fillings here to not confuse the legend numbering (HAVE TO GO FIRST --> UNDERLAY)
%%%%%%%%%%%%%%%%%%%%%%%%%%%%%%%%%%
%%%%%%%%%%%%%%%%%%%%%%%%%%%%%%%%%%
% actual curves start here!!
%%%%%%%%%%%%%%%%%%%%%%%%%%%%%%%%%%	
%%%%%%%%
		\addplot[NavyBlue,very thick,smooth] coordinates {
  (0,	1)
(0.002,	0.960353023)
(0.004,	0.931937444)
(0.006,	0.909363877)
(0.008,	0.891364358)
(0.01,	0.877294917)
(0.012,	0.866750877)
(0.014,	0.859448804)
(0.016,	0.855175605)
(0.018,	0.853762592)
(0.02,	0.855071092)
(0.022,	0.85898363)
(0.024,	0.865398228)
(0.026,	0.874224512)
(0.028,	0.885380985)
(0.03,	0.898793044)
(0.032,	0.914391434)
(0.034,	0.932111268)
(0.036,	0.951891039)
(0.038,	0.973672096)
(0.04,	0.997397978)

	};
		\addlegendentry{Sutter et al.~\cite{Sutter2014}}	
		
		\addplot[YellowOrange,very thick,smooth] coordinates {
 (0,	1)
(0.002,	0.962395349)
(0.004,	0.932817884)
(0.006,	0.906270007)
(0.008,	0.881684372)
(0.01,	0.858552392)
(0.012,	0.836573659)
(0.014,	0.815549173)
(0.016,	0.795337182)
(0.018,	0.775831488)
(0.02,	0.756949522)
(0.022,	0.738625221)
(0.024,	0.720804525)
(0.026,	0.703442366)
(0.028,	0.686500605)
(0.03,	0.669946542)
(0.032,	0.653751833)
(0.034,	0.637891679)
(0.036,	0.622344199)
(0.038,	0.607089947)
(0.04,	0.592111533)
	};
		\addlegendentry{Smith and Smolin~\cite{Smith2008b}}	
		
			%old UB
		\addplot[dashed,BrickRed,very thick,smooth] coordinates {
(0,	1)
(0.002,	0.95865033)
(0.004,	0.925730793)
(0.006,	0.896211439)
(0.008,	0.869140656)
(0.01,	0.843813132)
(0.012,	0.819793611)
(0.014,	0.796955738)
(0.016,	0.775166242)
(0.018,	0.754307326)
(0.02,	0.734284861)
(0.022,	0.715025582)
(0.024,	0.696464166)
(0.026,	0.678909421)
(0.028,	0.662110075)
(0.03,	0.645394424)
(0.032,	0.629203882)
(0.034,	0.613505867)
(0.036,	0.598683319)
(0.038,	0.584527498)
(0.04,	0.5701877)

	};
		\addlegendentry{new upper bound~\eqref{eq:private bound}}

%%%%%%%%%%%%%%%%%%%%%%%%%%%%%%%%%%%%%%%%%%%		
	%lower
		\addplot[dotted,ForestGreen,very thick,smooth] coordinates {
 (0,	1)
(0.002,	0.958371857)
(0.004,	0.924755279)
(0.006,	0.894169839)
(0.008,	0.86555691)
(0.01,	0.838413728)
(0.012,	0.81244418)
(0.014,	0.787452614)
(0.016,	0.763299977)
(0.018,	0.739882308)
(0.02,	0.717118915)
(0.022,	0.694945343)
(0.024,	0.673308914)
(0.026,	0.652165766)
(0.028,	0.631478813)
(0.03,	0.611216284)
(0.032,	0.591350658)
(0.034,	0.571857864)
(0.036,	0.552716671)
(0.038,	0.533908215)
(0.04,	0.515415622)
	};
	\addlegendentry{$Q^{(1)}$}

%%%%%%%%%%%%%%%%%%%%%%%%%%%%%%%%%%%%%%%%%%%		
	\end{axis}  
%%%%%%%%%%%%%%%%%%%%%%%%%%%%%%%%%%
% labeling here
%%%%%%%%%%%%%%%%%%%%%%%%%%%%%%%%%%

\end{tikzpicture}
 	\caption{Upper and lower bounds on the private capacity $P(\cB_{p,p})$ of the low-noise BB84 channel with $p_X=p_Y=p\in\left[0,0.04\right]$. The channel's coherent information $Q^{(1)}$ (green, dotted) provides a lower bound on the private capacity $P(\cB_{p,p})$. The red dashed line depicts our bound in Eq.~\eqref{eq:private bound}, which improves previously best known upper bounds via approximate degradable channels~\cite{Sutter2014} and additive extensions \cite{Smith2008b}.}
 	\label{fig:PN BB84 low noise}
\end{figure} 

\section{Discussions}
To summarize, we have established single-letter upper bounds on one-way entanglement distillation as well as quantum/private communication by considering extended states and channels, respectively.
We have seen that the concept of approximate degradable states (channels) and the idea of parametrized extended states (channels) can be combined to a more robust and general notion of approximate degradable extended states (channels). Our bounds have utilized the approximately preserved beneficial additivity properties of the approximate degradable states and channels as well as the flexibility or extra space provided by the flagged or extended systems. Our results have provided powerful and efficiently computable fundamental limits for quantum and private communication via quantum channels of both theoretical and practical interest, including the depolarizing channel, the BB84 channel, and the generalized amplitude damping channel. In particular, our improved upper bounds on the quantum and private capacities of these channels have deepened our understanding of quantum communication.
The codes for computing the bounds in this work can be found in \cite{Wang_codes}.

One interesting future direction is to explore the interaction between flagged channels and the degradable and anti-degradable decomposition of channels~\cite{Leditzky2017}, where the latter provides the state-of-art upper bounds on the quantum capacity of high-noise depolarizing channels. It is of great interest to have a further analytical study on the depolarizing channel using our approach with a similar spirit of \cite{Leditzky2018a} and to consider the behaviour using multiple flags. Another direction is to apply the main idea of this work to classical communication over quantum channels. For the classical capacity, there are several known upper bounds \cite{Bennett2002,Brandao2011c,Wang2017a,Wang2016g,Leditzky2018,Fang2019a} but the classical capacity of the amplitude damping channel is still open. It will be interesting to consider the interaction between extended channels and the approximate additivity of the Holevo information. 
%The concept of approximate degradable flagged channels may also be useful in network information theory due to the nice features of degradable broadcast channels~\cite{ElGamal2011}.

\section*{Acknowledgement}
We would like to thank Kun Fang, Felix Leditzky, and Kun Wang for discussions.  We also thank Mark M. Wilde for helpful suggestions and discussions.
 
\bibliographystyle{IEEEtran}

{\small \bibliography{bib}}
%\bibliography{bib}
%
\appendix
\section{SDP for the approximate degradability}\label{sec:app A}
\begingroup
 Note that for any two quantum channels $\cN_1, \cN_2$ from $A$ to $B$, the diamond norm of their difference can be expressed as a semidefinite program (SDP) of the following form~\cite{Watrous2009}: 		
\begin{align}\label{SDP diamond}
\frac12 \|\cN_1-\cN_2\|_\di = \inf \big\{\ \mu \ & \big|\ \tr_B Z \leq \mu \1_A,  Z \ge J_{\cN_1}-J_{\cN_2},\ Z \ge 0 \ \big\},
\end{align}
where $J_{\cN_1}$ and $J_{\cN_2}$ are the corresponding Choi-Jamio\l{}kowski matrices.
The degradability parameter of $\cN$ thus can be computed by the following SDP \cite{Sutter2014}:
\begin{subequations}\label{sdp: gN}
	\begin{align}
	\inf & \quad \ \mu\\
	\text{\rm s.t.} &\quad  \tr_{E} Z_{AE} \leq \mu I_{A}, Z_{AE}\ge 0,\\
		           & \quad  \tr_{E} J_{\cD}^{BE} = I_B, J_{\cD}^{BE}\ge 0, 
	 \quad \ Z_{AE} \ge  J_{\cN^c}^{AE}- J_{\cD\circ\cN}, \label{eq:condition compose}
	\end{align}
	where $\Phi_{EE'}$ is the unnormalized maximally entangled state with dimension $d_E$. Note that the last condition Eq.~\eqref{eq:condition compose} has semidefinite presentations,
	 (e.g., $\tr_{BE} (J_\cN^{AB}\otimes \Phi_{EE'})((J_{\cD}^{BE})^T\otimes I_{AE})$, see alternative presentations via transfer matrix in~\cite{Sutter2014}).
\end{subequations}

\endgroup
%%%%%%%%%%%%%%%%%%%%%%%%%%%%%%%
\section{Upper bound for the quantum capacity of the GAD channel}\label{sec:GAD bound}
\begingroup
We decompose the GAD channel as the following convex combination of two degradable channels:
\begin{align}
\cA_{y,N} = \sqrt{1-N} \cA^1_{y,N}(\rho) + \sqrt{N} \cA^2_{y,N}(\rho),
\end{align}
with $\cA^1_{y,N}(\rho) = A_1 \rho A_1^\dagger/(1-N)+A_2 \rho A_2^\dagger/(1-N)$ and  $\cA^2_{y,N}(\rho) = A_3 \rho A_3^\dagger/N+A_4 \rho A_4^\dagger/N$.
\begin{proposition}
For the GAD channel with noise parameter $y\in (0,1/2)$ and $N\in(0,1)$, it holds that
\begin{align}
Q({\cA}_{y,N}) &\le  Q(\widehat{\cA}_{y,N})= \max _{p \in[0,1]} I_{\mathrm{c}}\left(p\proj0+(1-p)\proj1, \widehat{\mathcal{A}}_{\gamma, N}\right), 
\end{align}
where $I_{\mathrm{c}}\left(\rho, \cN\right) \equiv H(\mathcal{N}(\rho))-H\left(\mathcal{N}^{c}(\rho)\right)$ and
$\widehat{\cA}_{y,N} = \sqrt{1-N} \cA^1_{y,N}(\rho) \otimes \proj{1} + \sqrt{N} \cA^2_{y,N}(\rho)\otimes \proj 0$.
\end{proposition}
\begin{proof}
	Since $\widehat{\cA}_{y,N}$ is the flagged channel of  ${\cA}_{y,N}$, we have 
	 \begin{align}
	 Q({\cA}_{y,N}) \le  Q(\widehat{\cA}_{y,N}).
	 \end{align}
	 Then, we have
	 \begin{align}
	 Q({\cA}_{y,N}) &\le   Q(\widehat{\cA}_{y,N})\\
	 & = Q^{(1)}(\widehat{\cA}_{y,N})\\
	 &= \max _{p \in[0,1]} I_{\mathrm{c}}\left(p\proj0+(1-p)\proj1, \widehat{\mathcal{A}}_{\gamma, N}\right). 
	 \end{align}
	 The first equality holds since $\widehat{\cA}_{y,N}$ is degradable. The second inequality is due to fact that the coherent information for degradable channels is a concave function, which implies that diagonal input states outperform non-diagonal states~\cite{Wolf2007}. Note that the covariance of this channel with respect to $\sigma_Z$ was used.
\end{proof}
\endgroup
\end{document}